%% file: arxiv.tex
\documentclass[a4paper,USenglish, fleqn,draft]{llncs}

\input{preamble}

\title{What are Strategies in  Delay Games?\newline Borel Determinacy for Games with Lookahead\thanks{Partially supported by the project ``TriCS'' (ZI~1516/1-1) of the German Research Foundation (DFG). The first author was supported by an IMPRS-CS PhD Scholarship.}}

\author{Felix Klein and Martin Zimmermann}

\institute{Reactive Systems Group, Saarland University, Germany\\
 \email{\{klein, zimmermann\}@react.uni-saarland.de}}

\begin{document}

\maketitle

\begin{abstract}
\input{abstract}
\end{abstract}

\section{Introduction}
\label{sec_intro}
\input{intro}
\section{Preliminaries}
\label{sec_defs}
\input{defs}

\subsection{Delay Games}
\label{subsec_delaygames}
\input{delaygames}

\subsection{The Borel Hierarchy}
\label{subsec_borel}
\input{borel}

%%%%%%%%%%%%%%%%%%%%%%%%%%%%%%%%%%%%%%%%%%%%%%%%%%%%%%%%%%%%%%%%%%%
%%%%%%%%%%%%%%%%%%%%%%%%%%%%%%%%%%%%%%%%%%%%%%%%%%%%%%%%%%%%%%%%%%%

\section{Borel Determinacy of Delay Games w.r.t.\ Fixed Delay Functions}
\label{sec_shiftresults}
\input{results_shift}

%%%%%%%%%%%%%%%%%%%%%%%%%%%%%%%%%%%%%%%%%%%%%%%%%%%%%%%%%%%%%%%%%%%
%%%%%%%%%%%%%%%%%%%%%%%%%%%%%%%%%%%%%%%%%%%%%%%%%%%%%%%%%%%%%%%%%%%

\section{Omnipotent Strategies in Delay Games}
\label{sec_strategies}
\input{strat}

\subsection{Omnipotent Strategies for Player~I}
\label{subsec_straI}
\input{stratI}

\subsection{Omnipotent Strategies for Player~O}
\label{subsec_straO}
\input{stratO}

%%%%%%%%%%%%%%%%%%%%%%%%%%%%%%%%%%%%%%%%%%%%%%%%%%%%%%%%%%%%%%%%%%%
%%%%%%%%%%%%%%%%%%%%%%%%%%%%%%%%%%%%%%%%%%%%%%%%%%%%%%%%%%%%%%%%%%%

\section{Borel Determinacy of Delay Games with Omnipotent Strategies}
\label{sec_skipresults}
\input{results_skip}

%%%%%%%%%%%%%%%%%%%%%%%%%%%%%%%%%%%%%%%%%%%%%%%%%%%%%%%%%%%%%%%%%%%
%%%%%%%%%%%%%%%%%%%%%%%%%%%%%%%%%%%%%%%%%%%%%%%%%%%%%%%%%%%%%%%%%%%

\section{Decidability}
\label{sec_dec}
\input{decidability}

%%%%%%%%%%%%%%%%%%%%%%%%%%%%%%%%%%%%%%%%%%%%%%%%%%%%%%%%%%%%%%%%%%%
%%%%%%%%%%%%%%%%%%%%%%%%%%%%%%%%%%%%%%%%%%%%%%%%%%%%%%%%%%%%%%%%%%%

\section{Characterizing the Existence of Omnipotent Strategies}
\label{sec_char}
\input{charac}

%%%%%%%%%%%%%%%%%%%%%%%%%%%%%%%%%%%%%%%%%%%%%%%%%%%%%%%%%%%%%%%%%%%
%%%%%%%%%%%%%%%%%%%%%%%%%%%%%%%%%%%%%%%%%%%%%%%%%%%%%%%%%%%%%%%%%%%

\section{Conclusion}
\label{sec_conc}
\input{conc}

\bibliographystyle{splncs03}
\bibliography{biblio}

\end{document}

%% file: preamble.tex
\usepackage{amsmath}
\usepackage{amssymb}

\usepackage[utf8]{inputenc}% erm\"oglich die direkte Eingabe der Umlaute 
\usepackage[T1]{fontenc} % das Trennen der Umlaute

\usepackage{todonotes}

\usepackage{wrapfig}
\usepackage{multirow}
\usepackage{todonotes}
\usepackage{booktabs}

\usepackage{tikz}
\usetikzlibrary{automata}
\usetikzlibrary{backgrounds}

\usepackage{pifont}
\usepackage{fancyvrb}

\newcommand{\myquot}[1]{``#1''}

\newcommand{\nats}{\mathbb{N}}
\newcommand{\natsplus}{\nats_+}
\newcommand{\size}[1]{|#1|}
\renewcommand{\epsilon}{\varepsilon}

\newcommand{\set}[1]{\{#1\}}
\newcommand{\widx}[2]{#1(#2)}

\newcommand{\aut}{\mathcal{A}}

\newcommand{\delaygame}[1]{\Gamma\!_{f}(#1)}
\newcommand{\delaygamep}[1]{\Gamma\!_{f'}(#1)}
\newcommand{\SigmaI}{\Sigma_I}
\newcommand{\SigmaO}{\Sigma_O}
\newcommand{\stratO}{\sigma}
\newcommand{\stratI}{\tau}
\newcommand{\p}{p}

\newcommand{\up}{\textsc{UP}}
\newcommand{\coup}{\textsc{co-UP}}
\newcommand{\np}{\textsc{NP}}
\newcommand{\conp}{\textsc{co-NP}}

\newcommand{\exptime}{\textsc{ExpTime}}

\SaveVerb{sharp}|#|

\newcommand{\skipp}{\mathrm{skip}}
\newcommand{\skippsym}{{\triangleright}}

\newcommand{\SigmaOskipp}{\SigmaO^{\skippsym}}
\newcommand{\shift}{\mathrm{shift}}

%% file: abstract.tex
\noindent We investigate determinacy of delay games with Borel winning conditions, infinite-duration two-player games in which one player may delay her moves to obtain a lookahead on her opponent's~moves.

First, we prove determinacy of such games with respect to a fixed evolution of the lookahead. However, strategies in such games may depend on information about the evolution. Thus, we introduce different notions of universal strategies for both players, which are evolution-independent, and determine the exact amount of information a universal strategy needs about the history of a play and the evolution of the lookahead to be winning. In particular, we show that delay games with Borel winning conditions are determined with respect to universal strategies. Finally, we consider decidability problems, e.g., \myquot{Does a player have a universal winning strategy for delay games with a given winning condition?}, for $\omega$-regular and $\omega$-context-free winning conditions.

%% file: intro.tex
Determinacy is the most fundamental property of a game: a game is determined, if one of the players has a winning strategy. One can even argue that a determinacy result paved the way for game theory: in 1913, Zermelo proved what is today known as Zermelo's theorem~\cite{Zermelo13}: every two-player zero-sum game of perfect information and finite duration is determined.

In this work, we are concerned with the infinite-duration variant of such games, so-called Gale-Stewart games. Such a game is played between Player~$I$ and Player~$O$ in rounds~$i \in \nats$: in round~$i$, Player~$I$ picks a letter $\alpha(i) \in \SigmaI$ and then Player~$O$ picks a letter \mbox{$\beta(i) \in \SigmaO$.} Player~$O$ wins, if the outcome~${\alpha(0) \choose \beta(0)}{\alpha(1) \choose \beta(1)}{\alpha(2) \choose \beta(2)} \cdots$ is in the winning condition~\mbox{$L \subseteq (\SigmaI \times \SigmaO)^\omega$}. Accordingly, a strategy for Player~$I$ is a function~$\stratI \colon \SigmaO^* \rightarrow \SigmaI$ mapping the previous moves of Player~$O$ to the next letter from $\SigmaI$ to be picked. The definition for Player~$O$ is dual. Note that a strategy cannot access the previous moves determined by itself. This is not a restriction, as they can always be reconstructed.

 Let $\rho(\stratI, \stratO)$ denote the outcome of the play where Player~$I$ employs the strategy $\stratI$ and Player~$O$ the strategy~$\stratO$. Then, determinacy can be characterized as follows: the negation~$\forall \stratO \exists \stratI.\, \rho(\stratI, \stratO) \notin L$ of $\exists \stratO \forall \stratI.\, \rho(\stratI, \stratO) \in L$ is equivalent to $\exists \stratI \forall \stratO.\, \rho(\stratI, \stratO) \notin L$, i.e., the order of the quantifiers can be swapped. 
 
Gale-Stewart games are an important tool in set theory and a long line of research into determinacy results for such games culminated in Martin's seminal Borel determinacy theorem~\cite{Martin75}: every Gale-Stewart game with a Borel winning condition is determined. On the other hand, using the axiom of choice, one can construct non-determined games. Even more so, determinacy of games with $\omega$-context-free conditions, which are not necessarily Borel, is equivalent to a large cardinal assumption that is not provable in ZFC~\cite{Finkel13}.

Gale-Stewart games also have important applications in theoretical computer science as they subsume games studied in automata theory, e.g., parity games and LTL realizability games. Showing the winning condition of a game to be Borel and then applying Martin's theorem is typically the simplest proof of determinacy for a novel winning condition. However, one can typically obtain stronger results, e.g., positional determinacy for parity games~\cite{EmersonJutla91,Mostowski91}. The quantifier swap induced by this determinacy result underlies (implicitly or explicitly) all complementation proofs for parity tree automata, the crucial step in proving decidability of monadic second-order logic over infinite trees. 
% 
%
%
%\begin{itemize}
%	\item History of Gale-Stewart Games: set theory, Banach-Mazur games (?)
%	\item Semi-formal definition, especially notion of strategy
%	\item History of determinacy results: discuss quantifier change and it's importance.
%	\item quickly mention trouble with determinacy for context-free conditions
%\end{itemize}

\medskip

\noindent\textbf{Delay Games.} Oftentimes, the strict alternation of moves in a Gale-Stewart game is too restrictive to model applications in computer science, e.g., in the presence of asynchronous components, buffers, or communication between components. Delay games, a relaxation of Gale-Stewart games, model such situations by allowing Player~$O$ to delay her moves in order to obtain a lookahead on her opponent's moves. This gives her an advantage and allows her to win games she would lose without lookahead. 

Furthermore, delay games have deep connections to uniformization problems for relations w.r.t.\ continuous functions~\cite{DBLP:conf/rex/ThomasL93,trakhtenbrot1973finite}. Consider a winning condition~$L \subseteq (\SigmaI \times \SigmaO)^\omega$: a winning strategy~$\stratO$ for Player~$O$ in a game with winning condition~$L$ induces a mapping $\lambda_\stratO \colon \SigmaI^\omega \rightarrow \SigmaO^\omega$ such that $\set{{\alpha \choose \lambda_\stratO(\alpha)} \mid \alpha \in \SigmaI^\omega} \subseteq L$: we say that $\lambda_\stratO$ uniformizes $L$. If $\stratO$ is winning for the Gale-Stewart game with winning condition~$L$, then $\lambda_\stratO$ is causal: the $n$-th letter of $\lambda_\stratO(\alpha)$ only depends on the first $n$ letters of $\alpha$. Furthermore, if $\stratO$ is winning in the delay game with winning condition~$L$ then $\lambda_\stratO$ is continuous in the Cantor topology. The latter result can even be refined: if $\stratO$ only delays moves a bounded number of times during each play, then $\lambda_\stratO$ is Lipschitz-continuous. Thus, uniformization problems w.r.t.\ (Lipschitz-)continuous functions are reducible to solving delay games. 

To capture and to analyze the precise amount of lookahead that is necessary to win, delay games are defined w.r.t.\ so-called delay functions, which represent the evolution of the lookahead. Thus, formally Player~$O$ does not decide to skip a move, but the delay function determines how many moves she skips: given a delay function~$f \colon \nats \rightarrow \natsplus$, the delay game~$\delaygame{L}$ is played in rounds, where in round~$i$ Player~$I$ has to pick $f(i)$ letters and afterwards Player~$O$ has to pick a single letter. Thus, if $f(i) > 1$, then Player~$O$'s lookahead increases by $f(i)-1$ letters. Typically, one is interested in the existence of a delay function~$f$ that allows Player~$O$ to win $\delaygame{L}$. One could imagine an alternative formalization where Player~$O$ may explicitly skip moves at her own choice. We will encounter this variant in Section~\ref{sec_skipresults}, where it is shown to be \emph{equivalent} to the one using delay functions.

Delay games where introduced by Hosch and Landweber who proved decidability of the existence of winning strategies with bounded lookahead for games with $\omega$-regular winning conditions~\cite{HoschLandweber72}. Later, Holtmann, Kaiser, and Thomas~\cite{HoltmannKaiserThomas12} proved that for such winning conditions, Player~$O$ has a winning strategy with bounded lookahead if and only if she has one with arbitrary lookahead, i.e., bounded lookahead always suffices for $\omega$-regular winning conditions. Furthermore, they gave a doubly-exponential upper bound on the necessary lookahead and a solution algorithm with doubly-exponential running time. These results were recently improved~\cite{KleinZimmermann14} by showing a tight exponential bound on the necessary lookahead and $\exptime$-completeness of solving delay games with $\omega$-regular winning conditions. Finally, delay games with deterministic $\omega$-context-free winning conditions are undecidable~\cite{FridmanLoedingZimmermann11}, while games with max-regular winning conditions w.r.t.\ bounded lookahead are decidable~\cite{Zimmermann14}. All these results can be expressed in terms of uniformization as well.

For all types of winning conditions mentioned above, delay games w.r.t.\ a fixed delay function are determined~\cite{FridmanLoedingZimmermann11,Zimmermann14}: these results are all ad-hoc as they rely on the existence of a deterministic automaton recognizing the winning condition and on determinacy of parity games on countable arenas: one can model the delay game as such a parity game where each vertex contains the whole history of the play as well as the state the automaton reaches when processing this history. 

%In terms of uniformization, Hosch and Landweber proved the decidability of the uniformization problem of $\omega$-regular relations by Lipschitz continuous functions, and Holtmann et al.\ proved the equivalence of the existence of a Lipschitz continuous uniformization function and the existence of a continuous uniformization function. Finally, uniformization for $\omega$-context-free relations is undecidable, but the Hosch-Landweber Theorem holds for max-regular relations. In further related work, Carayol, Löding, and Winter considered uniformization of relations over finite words~\cite{CarayolLoeding12}, finite trees~\cite{LoedingWinter14}, and infinite trees~\cite{CarayolLoeding07}.

%\begin{itemize}
%	\item Why delay games
%	\item semi-formal definition, discuss delay functions and connection to games where one player may skip whenever she wants to.
%	\item history of delay games
%	\item connection to uniformization, further results
%	\item determinacy results for delay games: all depend on deterministic (!) automata model and only hold for fixed delay function.
%\end{itemize}

\medskip

\noindent\textbf{What are Strategies in Delay Games?}
The most important aspect of a game are its (winning) strategies, e.g., in controller synthesis it is a winning strategy for the player representing the system that is turned into a controller. 

In a delay game, the notion of strategy is more complex than in a Gale-Stewart game due to the existence of the delay function: a strategy for Player~$I$ is of the form $\stratI \colon \SigmaO^* \rightarrow \SigmaI^*$ with $\size{\stratI(\beta(0) \cdots \beta(i-1))} =f(i)$, as he has to determine $f(i)$ letters in round~$i$. Thus, a strategy for Player~$I$ syntactically depends on $f$ and both players' strategies may depend semantically on $f$. On the one hand, this means that a winning strategy for a game w.r.t.\ a delay function~$f$ might not be applicable for an~$f' \neq f$. On the other hand, dependence on a particular delay function enables the reconstruction of the own previous moves. 

However, the classical definition of strategies for delay games introduced above is not useful when it comes to applications in synthesis\footnote{Nevertheless, the definition is still sufficient to study uniformization problems.}: the lack of robustness with regard to changes in the delay function is a serious problem. Furthermore, determinacy for delay games w.r.t.\ fixed delay functions is a rather unsatisfactory statement: for every $f$, either Player~$I$ has a winning strategy for $\delaygame{L}$ or Player~$O$ has one. If $\rho(f, \stratI, \stratO)$ denotes the outcome resulting from Player~$I$ employing $\stratI$ and Player~$O$ employing $\stratO$ in a game w.r.t.\ $f$, then the negation of $\exists \stratO \forall \stratI.\, \rho(f, \stratI, \stratO) \in L$ is equivalent to  $\exists \stratI \forall \stratO.\, \rho(f, \stratI, \stratO) \notin L$.  
However, the function~$f$ is quantified \emph{outside} of  the negation. 

Pushing the negation over the quantification of $f$ yields a much stronger statement, e.g., either there is an $f$ such that Player~$O$ wins $\delaygame{L}$ or Player~$I$ has a strategy that wins $\delaygame{L}$ w.r.t.\ every $f$. Note that such a strategy has to be universally applicable and winning for every $\delaygame{L}$ and may therefore neither syntactically nor semantically depend on a fix delay function. Thus, a  determinacy result w.r.t.\ such universal strategies means that the negation of $\exists f \exists \stratO \forall \stratI.\, \rho(f, \stratI, \stratO) \in L$ is equivalent to $\exists \stratI \forall f \forall \stratO.\, \rho(f, \stratI, \stratO) \notin L$, which is arguably a more natural notion.

%We propose several notions of universal strategies that differ in the amount of information about a play's history they can base their decision on and compare their strength in terms of games they are able to win. 

%
%
%\begin{itemize}
%	\item mention definition for classical games: reconstruction possible.
%	\item Discuss classical definition for delay game strategies: reconstruction possible, but depends on knowledge of $f$
%	\item problems with this definition
%	\item our proposal: universal strategies
%	\item question: which inputs: access to opponent's moves and/or information about lookahead.. yields different types of strategies: give short overview.
%\end{itemize}

\medskip

\noindent\textbf{Our Contribution.}
We study determinacy results for delay games with respect and without respect to fixed delay functions and with Borel winning conditions. 

Firstly, for games with fixed delay functions, we show determinacy w.r.t.\ classical strategies that may depend on the function under consideration. This result generalizes all previous determinacy results obtained via reductions to countable parity games using deterministic automata recognizing the winning condition. 

Secondly, we introduce universal strategies for delay games: for Player~$I$, we consider four variants that differ in the amount of information about a play's history they can access: the previous moves made by the strategy (which are not necessarily reconstructible) and the evolution of the lookahead in the previous rounds. We compare the strength of these strategies in terms of games they are able to win and show that they form a hierarchy whose first three levels are strict and that strategies in the fourth level are sufficient to win every game that is winable. It is open whether the inclusion between the last two levels is strict or not. For Player~$O$, we only consider two notions of universal strategies, as the second one is already sufficient to win every winable game. Furthermore, we show that the hierarchy for Player~$O$ is strict, too. 

Thirdly, we consider decision problems of the form \myquot{Does a player have a universal strategy for games with some given winning condition~$L$?}\ for $\omega$-regular and $\omega$-context-free winning conditions. We prove decidability (and tight complexity results) for both players in the $\omega$-regular case and for Player~$O$ i the deterministic $\omega$-context-free case. The other case and both cases for non-deterministic $\omega$-context-free winning conditions are undecidable. 

This work is meant as a starting point into the investigation of more general notions of strategies in delay games that are independent of the exact evolution of the lookahead and into  determinacy results w.r.t.\ these notions of strategies. We raise many open problems that are left open here. Most importantly, the exact amount of information about a play's history that is necessary to implement a universal strategy for Player~$I$ is open. Also, most of the decision problems remain open for the weaker notions of universal strategies we introduce. Finally, we expect there to be other natural notions of universal strategies for delay games, which might not have to be winning for every delay function (a very strong requirement), but for all $f$ for which the given player wins the delay game~$\delaygame{L}$ using classical strategies. 

%% file: defs.tex
The set of non-negative integers is denoted by $ \nats $ and we define $\natsplus = \nats \setminus
\set{ 0 }$. An
alphabet~$ \Sigma $ is a non-empty finite set, $ \Sigma^{*} 
$ ($ \Sigma^{n} $, $ \Sigma^{\omega} $) denotes the set of finite words (words of length~$ n $, infinite words) over $ \Sigma $. The
empty word is denoted by $ \epsilon $ and the length of a finite word~$ w $
by~$ \size{w} $. For $ w \in \Sigma^{*} \cup \Sigma^{\omega} $ and $n \in \nats$ we
write $ \widx{w}{n} $ for the $ n $-th letter of~$ w $.	

%% file: delaygames.tex
A delay function is a map $ f \colon \nats \rightarrow \natsplus  $. Given an $ \omega $-language $ L \subseteq
  \left( \SigmaI \times \SigmaO \right)^{\omega} $ and a delay
function~$ f $, the game $ \delaygame{L} $ is played by two players\footnote{For pronomial convenience~\cite{McNaughton00}, we assume Player~$I$ to be male and Player~$O$ to be female.},
the input player \myquot{Player~$ I $} and the output player \myquot{Player~$ O
$} in rounds $ i \in \nats $ as follows: in round~$ i $, Player~$ I $
picks a word $ u_{i} \in \SigmaI^{f(i)} $, then Player~$ O $ picks one
letter $ v_{i} \in \SigmaO $. We refer to the sequence $
(u_{0},v_{0}) (u_{1},v_{1}) (u_{2},v_{2}) \cdots$ as a play of $
\delaygame{L} $, which yields two infinite words $ \alpha =
u_{0}u_{1}u_{2} \cdots$ and $ \beta = v_{0}v_{1}v_{2} \cdots
$. Player~$ O $ wins the play if the outcome~$ {
  \widx{\alpha}{0} \choose \widx{\beta}{0} }{ \widx{\alpha}{1} \choose
  \widx{\beta}{1} }{ \widx{\alpha}{2} \choose \widx{\beta}{2} } \cdots
$ is in $ L $, otherwise Player~$ I $ wins.

Given a delay function $ f $, a strategy for Player~$ I $ is a mapping
$ \stratI \colon \SigmaO^{*} \rightarrow \SigmaI^{*} $ where $
\size{\stratI(w)} = f(\size{w}) $, and a strategy for Player~$ O $ is
a mapping $ \stratO \colon \SigmaI^{*} \rightarrow \SigmaO $. Consider a play~$ (u_{0},v_{0}) (u_{1},v_{1}) (u_{2},v_{2}) \cdots $ of $\delaygame{L} $. It is consistent with $ \stratI $, if $
u_{i} = \stratI(v_{0} \cdots v_{i-1}) $ for every $ i \in \nats $. It
is consistent with $ \stratO $, if $ v_{i} = \stratO(u_{0} \cdots
u_{i}) $ for every $ i \in \nats $. 

\begin{remark}
As usual, a strategy has only access to the opponents's moves, but not its own ones. However, this is not a restriction, since they can be reconstructed.

Fix a strategy~$\stratI$ for Player~$I$: in round~$i$, the input to $\stratI$ is the concatenation~$v_0 \cdots v_{i-1}$ of Player~$O$'s previous moves. The moves~$u_0, \ldots, u_{i-1}$ by Player~$I$ in the previous rounds are given by $u_{j} = \stratI(v_0 \cdots v_{j-1})$ for every $j < i$. 

Now, fix a strategy~$\stratO$ for Player~$O$: in round~$i$, the input to $\stratO$ is the concatenation~$u_0 \cdots  u_i$ of Player~$I$'s moves in the previous rounds, where each $u_j$ for $j \le i$ satisfies \mbox{$\size{u_j} = f(j)$}. Thus, the moves~$v_0, \ldots, v_{i-1}$ by Player~$O$ in the previous rounds are given by $v_{j} = \stratO(u_0 \cdots u_{j})$. 
Note this construction depends on knowledge about the delay function~$f$, as we decompose the input to $\stratO$ to obtain the prefix of length~$\sum_{j'=0}^{j} f(j')$.
\end{remark}

A strategy $ \stratI $ for Player~$p\in \set{I,O}$ is winning, if every play that is consistent with $\stratI$ is winning for Player~$\p$. We say that a player wins $ \delaygame{L} $, if she has a winning strategy and a delay game is determined, if one of the players wins it.
 
\begin{example} \label{example_introdelaygame}
Consider $ L_0 $ over $ \set{ a, b, c } \times \set{ b, c }
      $ with $ { \widx{\alpha}{0} \choose \widx{\beta}{0} }{
      \widx{\alpha}{1} \choose \widx{\beta}{1} }{ \widx{\alpha}{2}
      \choose \widx{\beta}{2} } \cdots \in L_0$ if $
    \widx{\alpha}{n} = a $ for every $ n \in \nats $ or if $
    \widx{\beta}{0} = \widx{\alpha}{n} $, where $ n $ is the smallest
    position with $ \widx{\alpha}{n} \neq a $. Intuitively, Player~$ O
    $ wins, if the letter she picks in the first round is equal to the
    first letter other than $ a $ that Player~$ I $
    picks. Also, Player~$ O $ wins, if there is no such
    letter. 

    We claim that Player~$ I $ wins $ \delaygame{L_0} $ for every delay function~$
    f $: he picks~$ a^{f(0)} $ in the first round and assume
    Player~$ O $ picks $ b $ afterwards (the case where she picks $ c
    $ is dual). Then, Player~$ I $ picks a word starting with $ c $ in
    the second round. The resulting play is winning for Player~$ I $
    no matter how it is continued. Thus, Player~$ I $ has a winning
    strategy in $ \delaygame{L_0} $.
\end{example}

Finally, we also consider delay-free games. Formally, these can be seen as delay games w.r.t.\ the delay function~$f$ with $f(i) = 1$ for every $i$, i.e., both players pick a single letter in each round. As $f$ is irrelevant, we denote such a game with winning condition~$L$ by $\Gamma(L)$.

%% file: borel.tex
\newcommand{\wadge}{W}
\newcommand{\bSigma}{\ensuremath{\mathbf{\Sigma}}}
\newcommand{\bPi}{\ensuremath{\mathbf{\Pi}}}

\newcommand{\ord}{\alpha}

Fix an alphabet $\Sigma$. The Borel hierarchy of $\omega$-languages over $\Sigma$ consists of levels $\bSigma_\ord$ and $\bPi_\ord$ for every countable ordinal~$\alpha > 0$, which are defined inductively by
\begin{itemize}
	\item $\bSigma_1 = \set{ L \subseteq \Sigma^\omega \mid L = K \cdot \Sigma^\omega \text{ for some } K \subseteq \Sigma^* }$,
	\item $\bPi_\alpha = \set{\Sigma^\omega \setminus L \mid L \in \bSigma_\alpha}$ for every $\alpha$, and
	\item $\bSigma_{\alpha} = \set{ \bigcup_{i \in \nats} L_i \mid L_i \in \bPi_{\alpha_i}  \text{ with } \alpha_i < \alpha \text{ for every } i}$ for every $\alpha >1$.
\end{itemize}

The following basic properties will be useful later on.

\begin{remark}
Let $\alpha$ be a countable ordinal.
\begin{itemize}

\item $\bSigma_\alpha \cup \bPi_\alpha \subseteq \bSigma_{\alpha+1} \cap \bPi_{\alpha+1}$.	

\item $\bSigma_\alpha$ and $\bPi_\alpha$ are closed under finite unions and finite intersections.

\end{itemize}

\end{remark}

A language~$L$ is Borel, if it is in one of the levels constituting the Borel hierarchy.

\begin{theorem}[Borel Determinacy Theorem~\cite{Martin75}]
Let $L$ be Borel. Then, $\Gamma(L)$ is determined.
\end{theorem}

%To prove membership of languages in the Borel hierarchy, one can employ continuous reductions. Here, we use a game theoretic definition via so-called Wadge games.
%
%Such a game $ \wadge(L, L') $ consists of two languages $ L
%\subseteq U^{\omega} $ and $ L' \subseteq V^{\omega} $ and is played
%in rounds between Player~I and Player~II. In each round $ i $, Player~I picks a letter $ x_{i} \in U $
%and Player~II picks a word $ y_{i} \in V^*$. After $ \omega $
%many rounds, the pair $ ( x, y ) $ with $ x =
%x_{0}x_{1}x_{2}... $ and $ y = y_{0}y_{1}y_{2}... $ determines the
%winner. Player~II wins if and only if $ y $ is infinite
%and we have $ x \in L $ iff $ y \in L' $.
%
%Here, a strategy for Player~I is a mapping $ \sigma \colon V^{*}
%\rightarrow U $ while a strategy for Player~II is a mapping $ \stratI \colon
%U^{+} \rightarrow V^{*} $. A play $ \rho = x_{0}y_{0}x_{1}x_{1}... $ is
%consistent with $ \sigma $ iff $ x_{i} = \sigma(y_{0}...y_{i-1}) $ for
%all $ i \in \nats $ and it is consistent with $ \stratI $ iff $ y_{i} =
%\stratI(x_{0}...x_{i}) $ for every $ i \in \nats $. A strategy is
%winning, if every play that is consistent with the strategy is winning.
%
%\begin{theorem}
%If $L'$ is Borel and Player~II has a winning strategy for the Wadge game~$\wadge(L, L')$, then $L$ is Borel, too.	
%\end{theorem}

%% file: results_shift.tex
Fix alphabets~$\SigmaI$ and $\SigmaO$ and a fresh skip symbol~$\skippsym \notin \SigmaO$, and define $\SigmaOskipp = \SigmaO\cup \set{\skippsym}$. To simplify our notation, let $h$ be the morphism that removes the skip symbol, i.e., the one defined by $h(\skippsym) = \epsilon$ and $h(a) =a $ for every $a \in \SigmaO$. Also, given two infinite words $\alpha$ and $\beta$ we write ${ \alpha \choose \beta}$ for the word ${\alpha(0) \choose \beta(0) } {\alpha(1) \choose \beta(1) }  {\alpha(2) \choose \beta(2) } \cdots $. Analogously, we write ${x \choose y}$ for finite words $x$ and $y$, provided they are of equal length.

Given a delay function~$f$ and an infinite word~$\beta \in \SigmaO^\omega$ we define $\shift_f(\beta) \in (\SigmaOskipp)^\omega$ by
\[
\shift_f(\beta) = \
\skippsym^{f(0)-1} \beta(0)
\skippsym^{f(1)-1} \beta(1)
\skippsym^{f(2)-1} \beta(2)
\cdots .
\] 
We lift this definition to languages~$L \subseteq (\SigmaI \times \SigmaO)^\omega$ via $\shift_f(L) = \set{ {\alpha \choose \shift_f(\beta)} \mid {\alpha \choose \beta} \in L }$. Intuitively, $\shift_f(L)$ encodes the delay function~$f$ explicitly by postponing Player~$O$'s moves using skip symbols. Thus, the delay game~$\delaygame{L}$ and the delay-free game~$\Gamma(\shift_f(L))$ are essentially equivalent: a winning strategy for Player~$\p \in \set{I,O}$ in $\delaygame{L}$ can directly be translated into a winning strategy for her in $\Gamma(\shift_f(L))$ and vice versa.

The main result of this section states that delay games with Borel winning conditions and w.r.t.\ fixed delay functions are determined.

\begin{theorem}
	Let $L$ be Borel and let $f$ be a delay function. Then, $\delaygame{L}$ is determined.
\end{theorem}

\begin{proof} 
We show that $\shift_f(L)$ is Borel. Then, our claim follows from the Borel determinacy theorem, as $\Gamma(\shift_f(L))$ and $\delaygame{L}$ are essentially the same game. 

We will prove the following statement, which is not the tightest result provable, but which suffices for our purposes: if $L \subseteq (\SigmaI \times \SigmaO)^\omega$ is in $\bSigma_\alpha$ (in $\bPi_\alpha$), then $\shift_f(L)$ is in $\bSigma_{\alpha+2}$ (in $\bPi_{\alpha+2}$). To this end, the language $U_f = \shift_f((\SigmaI \times \SigmaO)^\omega)$ will be useful. Note that $U_f$ contains exactly those plays during which the non-$\skipp$ symbols are played at the positions consistent with $f$. It is straightforward to show that $U_f$ is in $\bPi_2$.

First, assume we have $L \in \bSigma_1$, i.e., $L = K \cdot (\SigmaI \times \SigmaO)^\omega$ for some $K \subseteq (\SigmaI \times \SigmaO)^*$. Then, we have $\shift_f(L) = K' \cdot (\SigmaI \times \SigmaOskipp)^\omega \cap U_f$ where 
\begin{align*}
K' = \bigcup_{{\alpha(0) \choose \beta(0) } \cdots {\alpha(k) \choose \beta(k) } \in K} \bigg\{
{ x \choose y } \mid &
y = \skippsym^{f(0)-1} \beta(0) \cdots \skippsym^{f(k)-1} \beta(k) \text{ and }\\
& x \in \alpha(0) \cdots \alpha(k) \cdot \SigmaI^{|y| -(k+1)}
\bigg\},
\end{align*}
i.e., $\shift_f(L)$ is in $\bPi_2 \subseteq \bSigma_3$.

Now, let $L$ be in $\bPi_\alpha$, i.e., $L = (\SigmaI \times \SigmaO)^\omega \setminus L'$ for some $L' \in \bSigma_\alpha$. Applying the induction hypothesis yields that $\shift_f(L')$ is in $\bSigma_{\alpha+2}$. We have 
\[\shift_f(L) = ((\SigmaI \times \SigmaOskipp)^\omega \setminus \shift_f(L')) \cap U_f,\] i.e., $\shift_f(L) \in \bPi_{\alpha+2}$. 

Finally, assume we have $L \in \bSigma_{\alpha}$ for some $\alpha>1$, i.e., $L = \bigcup_{i \in \nats} L_i$ with $L_i \in \bPi_{\alpha_i}$ for some $\alpha_i < \alpha$. An application of the induction hypothesis shows that every $\shift_f(L_i)$ is in $\bPi_{\alpha_i+2}$.
Thus, $\shift_f(L) = \bigcup_{i \in \nats} \shift_f(L_i) $
is in $\bSigma_{\alpha+2}$ as $\alpha_i+2 < \alpha+2$ for every $i$. 
\end{proof}

Furthermore, from the equivalence of $\Gamma(\shift_f(L))$ and $\delaygame{L}$, which holds for arbitrary $L$, we obtain a result that is applicable to non-Borel winning conditions as well.

\begin{corollary}
If $\Gamma(\shift_f(L))$ is determined, then so is $\delaygame{L}$.	
\end{corollary}

%% file: strat.tex
In this section, we discuss different notions of strategies for delay games. The one introduced in Subsection~\ref{subsec_delaygames} is the classical one that was used in previous works~\cite{FridmanLoedingZimmermann11,HoltmannKaiserThomas12,KleinZimmermann14,Zimmermann14}. However, such strategies depend on a fixed delay function~$f$, i.e., they are not useful for a game w.r.t.\ a delay function~$f' \neq f$. This is a syntactic dependence in the case of Player~$I$, as he has to determine $f(i)$ letters in round~$i$. But even Player~$O$'s strategies may depend implicitly on knowledge about the delay function under consideration, as we will see below. 

In this section, we consider several stronger notions of universal strategies, i.e., strategies that are independent of the delay function under consideration. Informally, for Player~$I$ such a strategy returns an infinite word~$w \in \SigmaI^\omega$ and the first $f(i)$ letters of $w$ are used in round~$i$ of a delay game w.r.t.\ $f$. For Player~$O$, a universal strategy still returns a single letter, but it may no longer depend on information about the delay function under consideration. We say that a universal strategy is omnipotent for a winning condition~$L$, if the strategy is winning for every delay game~$\delaygame{L}$, independently of the choice of $f$.

\begin{example}
	\label{example_univstrat_naive}
Again, consider the winning condition~$L_0$ from Example~\ref{example_introdelaygame} and the strategy~$\stratI \colon \SigmaO^* \rightarrow \SigmaI^\omega$ given by $\stratI(\epsilon) = a^\omega$, $\stratI(bx) = c^\omega$, and $\stratI(cx) = b^\omega$ for $x \in \SigmaO^*$. Intuitively, in round~$0$, Player~$I$ can pick as many $a$'s as $f$ requires and then always picks $c$ (or $b$), if Player~$O$ has picked $b$ (or $c$) in round~$0$. This strategy is winning for him w.r.t.\ every $f$ and therefore omnipotent for $L_0$.
\end{example}

%% file: stratI.tex
We consider the following variants of universal strategies for Player~$I$, which differ in the amount of information about the play's history they base their decision on. In the example above, the strategy~$\stratI$ only has access to Player~$O$'s moves, which is sufficient to be winning. More powerful notions have (directly or indirectly) access to Player~$I$'s moves or to information about the delay function under consideration. Note that Player~$I$ cannot reconstruct his moves, if he only has access to Player~$O$'s moves, but not the delay function under consideration, which explains the need to access his own moves. 

\begin{enumerate}

	\item An \emph{output-tracking} (o.t.) strategy is a map~$\stratI \colon \SigmaO^* \rightarrow \SigmaI^\omega$.  Consider a play $ (u_{0},v_{0}) (u_{1},v_{1}) (u_{2},v_{2}) \cdots $ of $\delaygame{L} $ for some $f$: it is consistent with $\stratI$, if $u_i$ is the prefix of length $f(i)$ of $\stratI(v_0 \cdots v_{i-1})$. An o.t.\ strategy bases its decisions only on the moves~$v_j$  of Player~$O$ for $j \le i$ and can deduce the number of rounds already played, but has no way to reconstruct Player~$I$'s previous moves. In fact, it cannot even reconstruct the number of letters picked by Player~$I$ thus far.
			
	\item A \emph{lookahead-counting} (l.c.) strategy is a mapping~$\stratI \colon \SigmaO^* \times \nats \rightarrow \SigmaI^\omega$. This time, we say that a play~$ (u_{0},v_{0}) (u_{1},v_{1}) (u_{2},v_{2}) \cdots $ of $\delaygame{L} $ for some $f$ is consistent with $\stratI$, if $u_i$ is the prefix of length $f(i)$ of $\stratI(v_0 \cdots v_{i-1}, \sum_{j=0}^{i-1}f(j))$. A l.c.~strategy has access to the opponent's moves and the number of letters picked by Player~$I$ thus far. However, this still does not suffice for Player~$I$ to reconstruct the actual letters already picked.
	
	\item An \emph{input-output-tracking} (i.o.t.) strategy is a mapping~$\stratI \colon \SigmaO^* \times \SigmaI^* \rightarrow \SigmaI^\omega$. We define a play~$ (u_{0},v_{0}) (u_{1},v_{1}) (u_{2},v_{2}) \cdots $ of $\delaygame{L} $ for some $f$ to be consistent with $\stratI$, if $u_i$ is the prefix of $\stratI( u_0 \cdots u_{i-1}, v_0 \cdots v_{i-1})$  of length $f(i)$. An i.o.t.~strategy has access to both players' moves thus far, but cannot reconstruct when the moves of Player~$O$ were made.

	\item A \emph{history-tracking} (h.t.) strategy is a mapping~$\stratI \colon \SigmaO^* \times (\natsplus)^* \rightarrow \SigmaI^\omega$. Again, consider a play~$ (u_{0},v_{0}) (u_{1},v_{1}) (u_{2},v_{2}) \cdots $ of $\delaygame{L} $ for some $f$: it is consistent with $\stratI$, if $u_i$ is the prefix of length $f(i)$ of $\stratI(v_0 \cdots v_{i-1}, f(0) \cdots f(i-1))$. A h.t.~strategy has access to the opponent's moves and to the values of the delay function for all previous rounds, which allows him to reconstruct his moves. Thus, giving him additionally access to his previous moves does not increase the strength of such a strategy.

\end{enumerate}

As usual, we say that a strategy (of any type) is winning for Player~$I$ in $\delaygame{L}$ if the outcome of every play that is consistent with the strategy is in the complement of $L$. A strategy is omnipotent for $L$, if it is winning for Player~$I$ in $\delaygame{L}$ for every $f$.

The definitions above are given in order of increasing expressiveness, e.g., every (omnipotent) o.t.~strategy can be turned into an (omnipotent) l.c.~strategy inducing the same plays. The first two constructions are straightforward, and for the last one, Player~$I$ has to reconstruct his moves~$u_0, \ldots, u_{i-1}$ using the information about the values~$f(0), \ldots, f(i-1)$ and knowledge of his own i.o.t.~strategy.

Our first result shows that the first three types of strategies form a strict hierarchy in terms of the games that can be won with them. The last case will be discussed in Section~\ref{sec_char}: it is open whether omnipotent l.c.~strategies are strictly stronger than i.o.t.~strategies.

\begin{theorem}
\label{thm_strict_I}
There are winning conditions~$L_1$  and $L_2$ such that
\begin{enumerate}

	\item\label{thm_strict_I_1}
 Player~$I$ has an omnipotent l.c.~strategy for $L_1$, but no omnipotent o.t.~strategy, and

	\item\label{thm_strict_I_2}
 Player~$I$ has an omnipotent i.o.t.~strategy for $L_2$, but no omnipotent l.c.~strategy.

\end{enumerate}	
\end{theorem}

\begin{proof}
\ref{thm_strict_I_1}.) Let $L_1 = \set{ {\alpha \choose \beta} \mid \alpha \neq (ab)^\omega}$. Intuitively, Player~$I$  wins a game with winning condition~$L_1$ if he is able to produce the word~$(ab)^\omega$, the moves of Player~$O$ are irrelevant. Indeed, one can easily build a l.c.~strategy~$\stratI$ by defining $\stratI(x, n) = (ab)^\omega$ for even $n$ and $\stratI(x, n) = (ba)^\omega$ for odd $n$. Every outcome of a play that is consistent with $\stratI$ has $(ab)^\omega$ in its first component and is therefore winning for Player~$I$. Hence, $\stratI$ is omnipotent for $L_1$.
	
However, we claim that Player~$I$ has no omnipotent o.t.~strategy~$\stratI$ for $L_1$. If $\stratI(\epsilon) \neq (ab)^\omega$, then $\stratI$ is losing for some $f$ such that $f(0)$ is larger than the first position where  $\stratI(\epsilon)$ and $(ab)^\omega$ differ. Thus, we can assume $\stratI(\epsilon) = (ab)^\omega$. Now, fix some letter~$c \in \SigmaO$ and consider the first letter of $\stratI(c)$: if it is $a$ (the other case is dual), then $\stratI$ is losing for every $f$ with odd~$f(0)$, as the first component of the resulting outcome contains two $a$'s in a row, if Player~$O$ picks $c$ in the first round.

\ref{thm_strict_I_2}.) Fix $\SigmaI = \set{a, b, c}$ and $\SigmaO = \set{b,c}$, and define 
\[L_2 = (\SigmaI \times \SigmaO)^\omega \setminus \left\{ { \alpha \choose \beta} \mid \alpha \in  a^{n_0}\, \beta(0)\, a^{n_1}\, \beta(1)\, \cdot \SigmaI^\omega \text{ with } n_1 > n_0 \right\},\]
i.e., in order to win, Player~$I$ has to copy the first two letters picked by Player~$O$ and ensure to produce more $a$'s between these two positions than before the first one. It is straightforward to show that the following i.o.t.~strategy for Player~$I$ is omnipotent for $L_2$:
\[
\stratI(x,y) = \begin{cases}
 a^\omega &\text{if $x = \epsilon$,}\\
 x(0)\,a^\omega &\text{if $\size{x} =1$,}\\
 a^{n-k+1}\,x(1)\,a^\omega &\text{if $\size{x}>1$ and $y = a^n x(0) a^k$,}\\
 a^\omega &\text{otherwise.}\\
 \end{cases}
\]
Intuitively, Player~$I$ picks $a$'s until Player~$O$ has picked her first letter, which is immediately copied by Player~$I$. Then, he picks $a$'s until he has access to the second letter picked by Player~$O$ and then continues doing so until the second $a$-block is longer than the first one. Then, he copies Player~$O$'s second letter and only picks $a$'s afterwards. Note that it is crucial for Player~$I$ to have access to his own previous moves to guarantee that the second $a$-block is longer than the first one and that he does not necessarily pick $x(1)$ in the second round.

Next, we show that Player~$I$ has no omnipotent l.c.~strategy for $L_2$. Towards a contradiction, assume there is one, call it $\stratI$. We claim that $\stratI(x,n)$ begins with $aa$, $ax(0)$, or $x(0)a$ for every input~$(x,n)$ with $2\size{x} \le n$. 
This is straightforward for $x = \epsilon$ and $n=0$ (the cases $(\epsilon, n)$ with $n>0$ are irrelevant), since $\stratI(\epsilon, 0)$ has to be equal to $a^\omega$. If not, Player~$O$ would have a counterstrategy against $\stratI$ w.r.t.\ some large enough $f$ by picking $c$ in round~$0$, if the first non-$a$ letter in $\stratI(\epsilon, 0)$ is $b$ and vice versa.

It remains to consider an input~$(x,n)$ satisfying $0 < 2\size{x} \le n $. Consider a delay function~$f$ satisfying $f(0) = n - \size{x} +1$, $f(i) = 1$ for $i$ in the range~$1 \le i \le n$, and $f(n+1) =2$. Now, use $\stratI$ in $\delaygame{L_2}$ against Player~$O$ picking the letters of $x$ in the first $\size{x}$ rounds: Player~$I$ picks $a^{f(0)}$ in the first round, and $\alpha(1), \ldots, \alpha(\size{x}-1)$ during the next $\size{x}-1$ rounds, while Player~$O$ picks $x(0), \ldots, x(\size{x}-1)$. Note that we have $\size{a^{f(0)}\alpha(1) \cdots \alpha(\size{x}-1)} = n$; the next letters picked by Player~$I$ are therefore the first two of $\stratI(x, n)$. 

We consider two cases: if $\alpha(j) = a$ for every $j$, then the next two letters picked by $\stratI$ may contain only $a$'s, or one $a$ and the first letter of $x$, but not any other combination. Especially, the second letter of $x$ may not yet be picked, since the resulting $a$-block would be of length zero, which would results in a losing play. On the other hand, if there is a $j$ such that $\alpha(j) = x(0)$ then, all other $\alpha(j')$ have to be equal to $a$, since the second $a$-block would again be too short otherwise. Thus, the first $a$-block has at least length~$n-\size{x}+1$, which implies that the second $a$-block has at most length~$\size{x}-2$ after round $n$. As we have $n-\size{x}+1 \ge \size{x}+2$, we conclude that the next two letters picked by $\stratI$ have to be both an~$a$. 

To conclude, apply $\stratI$ in $\delaygamep{L_2}$ for the delay function~$f'$ with $f'(i)=2$ for every $i$ against Player~$O$ picking $b$ and $c$ in the first two rounds. As shown above, $\stratI$ will pick $aa$, $ab$, or $ba$ in every round. Thus, the resulting outcome is losing for Player~$I$, as he never picks a $c$. 
\end{proof}

Note that the winning condition~$L_1$ is $\omega$-regular and even recognizable by a deterministic $\omega$-automaton with reachability acceptance condition, and therefore in $\bSigma_1$. Furthermore, the winning condition~$L_2$ is not $\omega$-regular, but recognizable by a deterministic $\omega$-pushdown automaton with safety acceptance, and in $\bPi_1$.

%% file: stratO.tex
Now, we consider universal strategies for Player~$O$. The standard definition given in Subsection~\ref{subsec_delaygames} is syntactically independent of a fixed delay function. However, the reconstruction of Player~$O$'s moves made in previous rounds depends on knowledge about $f$. This can be exploited to show that strategies for Player~$O$ that have access to the number of rounds played already are more powerful than strategies which do not. Formally, we consider two types of omnipotent strategies for Player~$O$ corresponding to the first two notions for Player~$I$. The other two notions introduced for Player~$I$ are not necessary for Player~$O$.

\begin{enumerate}

	\item An \emph{input-tracking} (i.t.) strategy is a mapping~$\stratO \colon \SigmaI^* \rightarrow \SigmaO$. Consider a play $ (u_{0},v_{0}) (u_{1},v_{1}) (u_{2},v_{2}) \cdots $ of $\delaygame{L} $ for some $f$: it  is consistent with $\stratO$, if $v_i = \stratO(u_0 \cdots u_{i})$. Such a strategy cannot reconstruct Player~$O$'s previous moves and cannot even determine how many rounds were played already.
			
	\item A \emph{round-counting} (r.c.) strategy is a mapping~$\stratO \colon \SigmaI^* \times \nats \rightarrow \SigmaO$. This, time, we say that a play~$ (u_{0},v_{0}) (u_{1},v_{1}) (u_{2},v_{2}) \cdots $ of $\delaygame{L} $ for some $f$ is consistent with the strategy~$\stratO$, if $v_i=\stratO(u_0 \cdots u_{i}, i)$. A r.c.~strategy has access to the opponent's moves and the number of rounds played thus far. 

\end{enumerate}

Note the asymmetry between the counting strategies for Player~$I$ and Play\-er~$O$: Player~$I$ counts the number of letters he has picked thus far and therefore, as he has direct access to Player~$O$'s moves, the size of the lookahead. Player~$O$ counts the number of rounds, i.e., the number of letters she has picked thus far. Again, this allows her to determine the size of the lookahead, as she has access to Player~$I$'s moves. 
Omnipotency for Player~$O$'s strategies is defined as before. Also, as for Player~$I$, every i.t.~strategy can be turned into an r.c.~strategy. Finally, r.c.~strategies are more powerful than i.t. ones. 

\begin{theorem}
\label{thm_strict_O}
There is a winning condition~$L_3$ such that Player~$O$ has an omnipotent r.c.~strategy for $L_3$, but no omnipotent i.t.~strategy.
\end{theorem}

The proof is a variation of the analogue for Player~$I$: the winning condition requires Player~$O$ to produce $(ab)^\omega$, which she can do, if she has access to the number of rounds already played, but she cannot do it without this information. Again, the distinguishing winning condition is very simple: it is $\omega$-regular and even recognizable by a deterministic $\omega$-automaton with safety acceptance condition, and therefore in $\bPi_1$. 

%% file: results_skip.tex
Now, we turn our attention to delay games without fixed delay functions and show that there is either a winning strategy for Player~$O$ for some $f$, or Player~$I$ wins for every $f$ with the \emph{same} omnipotent strategy. Then, we show the dual result for Player~$O$.

We still use the notation introduced at the beginning of Section~\ref{sec_shiftresults} and start by defining $\skipp(L) = \bigcup_f \shift_f(L)$, where there union ranges over all delay functions~$f$. Note that Player~$O$ loses a play in a game with winning condition~$\skipp(L)$ if she picks $\skippsym$ all but finitely often. Also, we have $\skipp(L) = \set{ {\alpha \choose \beta } \in  (\SigmaI \times \SigmaOskipp)^\omega \mid  {\alpha \choose h(\beta) } \in L }$.

The tight connection between delay games~$\delaygame{L}$ for arbitrary~$f$  and the delay-free game~$\Gamma(\skipp(L))$ appears implicitly in the work by Holtmann et al.~\cite{HoltmannKaiserThomas12} and is made explicit below. We exploit these connections to prove determinacy of delay games with Borel winning conditions.

\begin{theorem}
\label{thm_universal_I}
	Let $L$ be Borel. Either, Player~$O$ wins $\delaygame{L}$ for some $f$ or Play\-er~$I$ has an omnipotent h.t.~strategy for $L$. 
\end{theorem}

\begin{proof}
	First, we show that $\skipp(L)$ is Borel and then apply the connection between the games~$\delaygame{L}$ for arbitrary $f$ and $\Gamma(\skipp(L))$.
	
Proving $\skipp(L)$ to be Borel is analogous to the proof for $\shift_f(L)$, we just replace the intersections with $U_f$ by intersections with $U = \skipp((\SigmaI \times \SigmaO)^\omega)$ (which is also in $\bPi_2$) and the definition of $K'$ in the induction start is changed to 
\begin{align*}
K' = \bigcup_{{\alpha(0) \choose \beta(0) } \cdots {\alpha(k) \choose \beta(k) } \in K} \bigg\{
{ x \choose y } \mid &
y \in \skippsym^* \beta(0) \cdots \skippsym^* \beta(k) \text{ and }\\
& x \in \alpha(0) \cdots \alpha(k) \cdot \SigmaI^{|y| -(k+1)}
\bigg\},
\end{align*}

Thus, $\Gamma(\skipp(L))$ is determined. 

Next, we show that Player~$O$ wins $\delaygame{L}$ for some $f$, if she wins $\Gamma(\skipp(L))$. Let $\stratO'$ be a winning strategy for Player~$O$ in $\Gamma(\skipp(L))$. We construct a delay function~$f$ and a winning strategy~$\stratO$ for Player~$O$ in $\delaygame{L}$ by simulating a play in $\delaygame{L}$ by a play in $\Gamma(\skipp(L))$.

We begin by defining $f$. For $i \in \nats$ let $\ell_i$ be the maximal number such that Player~$O$ picks at most $i$ non-skip symbols during the first $\ell_i$ rounds in every play of $\Gamma(\skipp(L))$ that is consistent with $\stratO'$. We claim that every $\ell_i$ is well-defined. Assume $\ell_i$ for some fixed~$i$ is not. Then, the play prefixes under consideration for defining $\ell_i$ form an infinite, but finitely branching tree. Hence, König's Lemma implies the existence of an infinite play that is consistent with $\stratO'$ during which Player~$O$ all but finitely often picks $\skippsym$. This play is losing for her, thus contradicting $\stratO'$ being a winning strategy. 

By construction, if Player~$I$ has picked $\ell_i +1$ letters in $\Gamma(\skipp(L))$, then $\stratO'$ has determined at least $i+1$ non-skip letters. Now, let $f(0) = \ell_0 +1$ and $f(i+1) = ( \ell_{ i+1 } +1 ) - \sum_{j = 0}^{i} f(j)$.  It remains to define $\stratO$: assume Player~$I$ has picked $u_0, \ldots , u_i$ in rounds~$i = 0,1, \ldots, i$ with $\size{u_j} = f(j)$. Consider the play prefix in $\Gamma(\skipp(L))$ during which Player~$I$ picks $u_0 \cdots u_i$ and Player~$O$ plays according to $\stratO'$. We define $\stratO(u_0 \cdots u_i)$ to be the $i$-th non-skip letter (starting with the $0$-th letter) picked by Player~$O$ on this play prefix. This is well-defined by the definition of $f$.

Let ${ \alpha \choose \beta }$ an outcome that is consistent with $\stratO$. A straightforward induction shows that there is a play in $\Gamma(\skipp(L))$ that is consistent with $\stratO'$ and has an outcome~${ \alpha \choose \beta' }$ such that $\beta = h(\beta')$. Hence, $\stratO'$ being a winning strategy implies ${ \alpha \choose \beta' } \in \skipp(L)$ and therefore ${ \alpha \choose \beta } \in L$. Thus, $\stratO$ is a winning strategy for Player~$O$ in $\delaygame{L}$.

To conclude, we show that Player~$I$ has an omnipotent h.t.~strategy~$\stratI$ for $L$, if he wins $\Gamma(\skipp(L))$. To this end, let $\stratI' \colon (\SigmaOskipp)^* \rightarrow \SigmaI$ be a winning strategy for Player~$I$ in $\Gamma(\skipp(L))$.
We define an h.t.~strategy~$\stratI \colon \SigmaO^* \times (\natsplus)^* \rightarrow \SigmaO^\omega$. Let $x \in \SigmaO^*$ and $n_0 \cdots n_{i-1} \in (\natsplus)^*$. Note that $\stratI$ will only be applied to inputs~$(x, n_0 \cdots, n_{i-1})$ where $\size{x} = i$. Thus, we restrict our attention to those inputs. Let
\[x' = \skippsym^{n_0-1} \, x(0) \, \skippsym^{n_1-1} \, x(1) \, \cdots \skippsym^{n_{i-1}-1} \, x(i-1) \in (\SigmaOskipp)^* \]
and define
\[
\stratI(x, n_0 \cdots n_{i-1}) = 
\stratI'(x')\, 
\stratI'(x'\skippsym)\, 
\stratI'(x'\skippsym\skippsym)\,
\stratI'(x'\skippsym\skippsym\skippsym)\, \cdots
,
\]
i.e., the answers according to $\stratI'$ to Player~$O$ picking $\skippsym$ ad infinitum after picking $x'$.

A straightforward induction shows that for every outcome~${\alpha \choose \beta}$ that is consistent with $\stratI$ in $\delaygame{L}$ for some $f$, there is an outcome~${\alpha \choose \beta'}$ that is consistent with $\stratI'$ such that $h(\beta') = \beta$. As $\stratI'$ is winning for Player~$I$ in $\Gamma(\skipp(L))$ we have ${\alpha \choose \beta'} \notin \skipp(L)$ and thus ${\alpha \choose \beta} \notin L$. Hence, $\stratI$ is winning for $\delaygame{L}$ for every $f$ and therefore omnipotent for $L$.
 \end{proof}

The second part of the proof above (the equivalence of the delay games and the delay-free game) works for arbitrary winning conditions. Hence, we obtain the following corollary.

\begin{corollary}
If $\Gamma(\skipp(L))$ is determined, then either Player~$O$ wins $\delaygame{L}$ for some $f$ or Player~$I$ has an omnipotent h.t.~strategy for $L$. 
\end{corollary}

It is open whether these results hold for i.o.t.~strategies as well. This is related to the strictness of the strategy hierarchy mentioned earlier: is there a winning condition~$L$ such that Player~$I$ has an omnipotent h.t.~strategy for $L$, but no omnipotent i.o.t.~strategy? We discuss this question in Section~\ref{sec_char}.

To conclude this section, let us consider the case where Player~$O$ wins $\delaygame{L}$ for every delay function~$f$. Here, we apply a monotonicity argument: the larger the lookahead is, the more information Player~$O$ has at her disposal, which makes winning easier for her. Thus, if she wins w.r.t.\ every delay function, then she wins in particular without lookahead. A winning strategy for the delay-free game can be turned into a winning strategy w.r.t.\ every \textit{larger} delay function. Thus, the omnipotent strategy for Player~$O$ mimics the behavior of a winning strategy for the delay-free game and ignores the additional information given by the lookahead.

Formally, we order delay functions by the amount of lookahead available for Player~$O$ at every round: we define $f \sqsubseteq f'$, if and only if $\sum_{j=0}^i f(j) \le \sum_{j=0}^i f'(j)$ for every $i \in \nats$, i.e., in every round, the lookahead granted by $f'$ is at least as large as the one granted by $f$. A winning strategy for Player~$O$ w.r.t.\ $f$ can easily be turned into one for $f'$ by ignoring the additional information. Thus, we obtain the following monotonicity property.

\begin{remark}
If $f \sqsubseteq f'$ and Player~$O$ wins $\delaygame{L}$, then also $\delaygamep{L}$.
\end{remark}

Note that \emph{winning} refers to winning strategies that may depend on $f$ respectively $f'$. Nevertheless, we can use monotonicity to obtain omnipotent strategies by considering the $\sqsubseteq$-minimal delay function. More formally, if Player~$O$ wins $\delaygame{L}$ for every $f$, then she wins in particular the delay-free game~$\Gamma(L)$, i.e., the game w.r.t.\ the $\sqsubseteq$-minimal delay-function~$i \mapsto 1$. It is easy to see that a winning strategy~$\stratO'$ for Player~$O$ in $\Gamma(L)$ can be turned into an omnipotent r.c.~strategy~$\stratO$ for $L$: defining $\stratO(x,i) = \stratO'(x(0)\cdots x(i-1))$ simulates the strategy~$\stratO'$ by ignoring the additional information gained due to the lookahead.

\begin{theorem}
\label{thm_universal_O}
Either, Player~$I$ wins $\delaygame{L}$ for some $f$ or Player~$O$ has an omnipotent r.c.~strategy for $L$. 
\end{theorem}

As shown in Theorem~\ref{thm_strict_O}, such a strategy has to have access to the number of rounds already played, the theorem does not hold for input-tracking strategies. Note however, that this results holds for arbitrary winning conditions~$L$.

A similar construction works if Player~$O$ does not have an omnipotent strategy for $L$, but wins $\delaygame{L}$ for some $f$. Then, she can simulate a winning strategy for $\delaygame{L}$ in $\delaygamep{L}$ for every $f' \sqsupseteq f$.

%% file: decidability.tex
In this section, we consider decision problems regarding omnipotent strategies, i.e., we are interested in determining whether a given player has an omnipotent strategy for a given winning condition~$L$. 

We begin with $\omega$-regular conditions represented by deterministic parity automata. 

\begin{theorem}\label{thm_dec_parity}
The following problems are $\exptime$-complete respectively in\newline $\np \cap \conp$:
\begin{enumerate}

	\item\label{thm_dec_parity_I}
 Given a deterministic parity automaton~$\aut$, does Player~$I$ have an omnipotent h.t.~strategy for $L(\aut)$?

	\item\label{thm_dec_parity_O} 
 Given a deterministic parity automaton~$\aut$, does Player~$O$ have an omnipotent r.c.~strategy for $L(\aut)$?
	
\end{enumerate} 
\end{theorem}

\begin{proof}
\ref{thm_dec_parity_I}.) Due to Theorem~\ref{thm_universal_I}, Player~$I$ has an omnipotent h.t.~strategy for $L(\aut)$ if and only if there is no $f$ such that Player~$O$ wins $\delaygame{L(\aut)}$. 
Determining whether there is an $f$ such that Player~$O$ wins $\delaygame{L(\aut)}$ is $\exptime$-complete~\cite{KleinZimmermann14}. Hence, determinacy of $\omega$-regular delay games w.r.t.\ fixed delay functions and closure of $\exptime$ under complements yields the desired result. 

\ref{thm_dec_parity_O}.) Due to Theorem~\ref{thm_universal_O}, Player~$O$ has an omnipotent r.c.~strategy for $L(\aut)$ if and only if she wins $\Gamma(L(\aut))$. This game can be encoded as a parity game in an arena of size~$2\size{\aut}$ that has the same colors as $\aut$. The winner of this game is solvable in $\np \cap \conp$ (and even $\up \cap \coup$~\cite{Jurdzinski98}), which yields the desired result.
\end{proof}

An omnipotent r.c.~strategy for Player~$O$ can be implemented by a finite automaton with output of size~$\mathcal{O}(\size{\aut})$, e.g., if the input~$(x,n) \in \SigmaI^* \times \nats$ with $\size{x} \ge n$ is encoded as $x(0) \cdots x(n-1) \# x(n) \cdots x(\size{x}-1)$, where $\#$ is a fresh symbol. The states of the automaton are the vertices of the parity game constructed in the proof above and the output function is given by a positional winning strategy for this game. 

Now, we turn our attention to $\omega$-context-free winning conditions. Such languages are recognized by $\omega$-pushdown automata, classical pushdown-automata running on infinite words. We refer to~\cite{CohenG78} for detailed definitions. First, we consider deterministic automata.

\begin{theorem}\label{thm_dec_detpd}
The following problems are undecidable respectively $\exptime$-complete:
\begin{enumerate}

	\item\label{thm_dec_detpd_I} Given a deterministic $\omega$-pushdown automaton~$\aut$, does Player~$I$ have an omnipotent h.t.~strategy for $L(\aut)$?

	\item\label{thm_dec_detpd_O} Given a deterministic $\omega$-pushdown automaton~$\aut$, does Player~$O$ have an omnipotent r.c.~strategy for $L(\aut)$?
	
\end{enumerate} 
\end{theorem}

\begin{proof} Recall that delay games with 
winning conditions that are  recognized by deterministic \mbox{$\omega$-pushdown} automata and w.r.t.\ fixed delay functions are determined~\cite{FridmanLoedingZimmermann11}.

\ref{thm_dec_detpd_I}.) As in the $\omega$-regular case, Player~$I$ has an omnipotent h.t.~strategy for $L(\aut)$ if and only if there is no $f$ such that Player~$O$ wins $\delaygame{L(\aut)}$. 
Determining whether there is an $f$ such that Player~$O$ wins $\delaygame{L(\aut)}$ is undecidable~\cite{FridmanLoedingZimmermann11}. Hence, determinacy w.r.t.\ fixed delay functions implies undecidability of the problem. 

\ref{thm_dec_detpd_O}.) Again, as in the $\omega$-regular case, the problem can be reduced to solving the delay-free game~$\Gamma(L(\aut))$, which is $\exptime$-complete~\cite{Walukiewicz01}.
\end{proof}

As before, an omnipotent r.c.~strategy for Player~$O$ can be represented finitely by constructing a pushdown-automaton with output that implements a winning strategy for the delay-free game~$\Gamma(L(\aut))$, which can be constructed effectively~\cite{Walukiewicz01}.

To conclude, we consider non-deterministic $\omega$-pushdown automata.

\begin{theorem}\label{thm_dec_nondetpd}
The following problems are undecidable:
\begin{enumerate}

	\item\label{thm_dec_nondetpd_I} Given a non-deterministic $\omega$-pushdown automaton~$\aut$, does Player~$I$ have an omnipotent l.c. (i.o.t., h.t.)~strategy for $L(\aut)$?

	\item\label{thm_dec_nondetpd_O} Given a non-deterministic $\omega$-pushdown automaton~$\aut$, does Player~$O$ have an omnipotent r.c.~strategy for $L(\aut)$?
	
\end{enumerate} 
\end{theorem}

\begin{proof} Recall that the (non-)universality problem for non-deterministic $\omega$-push\-down automata is undecidable~(see, e.g.,~\cite{Finkel01}). Given such an automaton~$\aut$, we define the winning condition~$I_{\!\aut} = \set{{\alpha \choose \alpha} \mid \alpha \in L(\aut)}$, i.e., in order to win, Player~$I$ has to produce an $\alpha \notin L(\aut)$.

\ref{thm_dec_nondetpd_I}). We prove undecidability for the case of l.c.~strategies by a reduction from the non-universality problem. The other cases are proven similarly.

We claim that $L(\aut)$ is non-universal if and only if Player~$I$ has an omnipotent l.c.~strategy for $I_{\!\aut}$. Let $L(\aut)$ be non-universal, i.e., we can fix some $\alpha \notin L(\aut)$. Then, there is a l.c.~strategy for Player~$I$ that produces $\alpha$, independently of the moves of Player~$O$. Hence, this strategy is omnipotent for $I_{\!\aut}$.
Now, if $I_{\!\aut}$ is universal then Player~$O$ wins $\Gamma(I_{\!\aut})$ by just copying Player~$I$'s moves. Hence, Player~$I$ has no omnipotent strategy for $I_{\!\aut}$.

\ref{thm_dec_nondetpd_O}.) We claim that $L(\aut)$ is universal if and only if Player~$O$ has an omnipotent r.c.~strategy for $I_{\!\aut}$. Above, we have shown that $I_{\!\aut}$ being universal implies that Player~$O$ wins $\delaygame{I_{\!\aut}}$ for every $f$. Hence, due to Theorem~\ref{thm_universal_O}, Player~$O$ has an omnipotent r.c.~strategy for $I_{\!\aut}$. On the other hand, as seen above, if $I_{\!\aut}$ is non-universal, then Player~$I$ has an omnipotent strategy for $I_{\!\aut}$, which implies that $I_{\!\aut}$ has none. 
\end{proof}

It is open whether the problems asking for weaker types of omnipotent strategies are decidable. We discuss these problems in the next section.

%% file: charac.tex
In this section, we give a characterization of omnipotent strategies for delay games in terms of uniform strategies for delay-free games. We focus on the case of i.o.t.~strategies for Player~$I$, but the other cases are analogous. 

Fix some strategy~$\stratI \colon (\SigmaOskipp)^* \rightarrow \SigmaI$ and define the equivalence relation~$\approx_\stratI$ over $(\SigmaOskipp)^*$ via $x_0 \approx x_1$ if and only if $\size{x_0} = \size{x_1}$, $h(x_0) = h(x_1)$, and $\stratI(x_0') = \stratI(x_1')$ for all proper prefixes $x_0' \sqsubset x_0$ and $x_1' \sqsubset x_1$ with $\size{x_0'} = \size{x_1'}$. Thus, $x_0$ and $x_1$ are equivalent, if Player~$O$ has picked the same sequence of non-skip symbols in $x_0$ and in $x_1$, has picked the same number of skip symbols (but possibly at different positions), and $\stratI$ picked the same moves answering to Player~$O$ picking $x_0$ and $x_1$, respectively, during the previous rounds. 

Now, we say that a strategy~$\stratI$ for Player~$I$ in $\delaygame{\skipp(L)}$ is i.o.-uniform if 
$\stratI(x) = \stratI(x')$ for all $x \approx_\stratI x'$. The following lemma is a straightforward extension of Theorem~\ref{thm_universal_I}.

\begin{lemma}
\label{lem_uniform_iotrack}
Player~$I$ has an omnipotent i.o.t.~strategy if and only if Player~$I$ has an i.o.-uniform winning strategy for $\Gamma(\skipp(L))$.
\end{lemma}

We conjecture that Player~$I$ always has such a uniform strategy. 

\begin{conjecture}
\label{conj}
If Player~$I$ wins $\Gamma(\skipp(L))$, then she has an i.o.-uniform winning strategy for $\Gamma(\skipp(L))$.	
\end{conjecture}

Note that we do not impose any requirements on $L$. If the conjecture is true, then Theorem~\ref{thm_universal_I} is also true for i.o.t.~strategies.   

The existence of an omnipotent o.t., l.c., or i.t.~strategy for a winning condition~$L$ can be characterized analogously using appropriate equivalence relations that capture the limited access to information about the history of a play that such a strategy has. 

Furthermore, the existence of such uniform strategies can be expressed in the framework introduced by Maubert and Pinchinat~\cite{MaubertP14}: they investigate  infinite games under uniformity constraints on strategies expressed in an extension of LTL with a modality to equate finite play prefixes that are in some given equivalence relation. The logic is able to express the uniformity constraint formulated above, but our problems are not in the decidable fragment presented in this work, as the equivalence relations that characterize universal strategies are not rational (recognizable by an asynchronous transducer) and turning an $\omega$-regular~$L$ into $\skipp(L)$ does not preserve $\omega$-regularity.

%% file: conc.tex
We presented determinacy results for delay games with Borel winning conditions, both with and without respect to fixed delay functions: in the latter case, we showed the existence of omnipotent strategies, i.e., strategies that are winning w.r.t.\ every delay function. In particular, we analyzed the exact amount of information such a strategy needs about the history of the play and the delay function under consideration. For games w.r.t.\ a fixed delay function, on which winning strategies may depend, access to the opponent's moves is sufficient. However, for omnipotent strategies the situation is more intricate: Player~$O$ needs access to the opponent's moves and the number of rounds played thus far, just having access to the opponent's moves is not sufficient. For Player~$I$, we showed that access to both player's moves is necessary and having the full information about the play's history is trivially sufficient. However, it is open whether that much information is necessary: does access to both player's previous moves, but not to the delay function under consideration, suffice to implement an omnipotent strategy? To answer this question, we currently work on resolving Conjecture~\ref{conj}.

Also, we determined the precise computational complexity of decision problems of the following form for $\omega$-regular and $\omega$-context-free winning conditions: given a winning condition $L$, does Player~$p \in \set{I,O}$ have an omnipotent strategy for $L$? 

Another interesting question concerns the decision problems left open in Section~\ref{sec_dec}: can one decide if Player~$I$ has an omnipotent o.t.\ (l.c., i.o.t.)~strategy for a given $\omega$-regular winning condition? The analogous question for Player~$O$ and input-tracking strategies is also open. Furthermore, we left open the finite representability of omnipotent strategies for Player~$I$ for $\omega$-regular winning conditions. We expect the techniques we developed to give an exponential-time algorithm for solving $\omega$-regular delay games~\cite{KleinZimmermann14} to yield such strategies, but this is beyond the scope of this paper. 

Another interesting open problem is to develop a theory of finite-state and positional winning strategies for delay games, both for the case with a fixed delay function and the universal case, and to prove positional respectively finite-state determinacy results.\medskip

\noindent\textbf{Acknowledgments.} The work presented here was initiated by a discussion with Dietmar Berwanger at the Dagstuhl Seminar \myquot{Non-Zero-Sum-Games and Control} in 2015.